\documentclass[10pt]{amsart}
\usepackage{amsmath,amsfonts,amssymb,mathrsfs}

\title[The Boltzmann equation]{\rm The spatially homogeneous relativistic
Boltzmann equation with a hard potential}

\author{Ho Lee and Alan D. Rendall}

\begin{document}

\newtheorem{theorem}{Theorem}[section]
\newtheorem{lemma}{Lemma}[section]
\newtheorem{corollary}{Corollary}[section]
\newtheorem{proposition}{Proposition}[section]
\newtheorem{remark}{Remark}[section]
\newtheorem{definition}{Definition}[section]

\renewcommand{\theequation}{\thesection.\arabic{equation}}
\renewcommand{\thetheorem}{\thesection.\arabic{theorem}}
\renewcommand{\thelemma}{\thesection.\arabic{lemma}}
\newcommand{\bbr}{{\mathbb R}}
\newcommand{\bbs}{{\mathbb S}}

\begin{abstract}
In this paper, we study spatially homogeneous solutions of the Boltzmann
equation in special relativity and in Robertson-Walker spacetimes. We obtain
an analogue of the Povzner inequality in the relativistic case and use it to
prove global existence theorems. We show that global solutions exist for a
certain class of collision cross sections of the hard potential type in
Minkowski space and in spatially flat Robertson-Walker spacetimes.
\end{abstract}

\maketitle

\section{Introduction}
Under the assumption of homogeneity the non-relativistic Boltzmann equation has
been extensively studied for many years. Homogeneity means that the unknown
in the equation, which in general depends on time, spatial variables and
velocity variables, is restricted to depend only on time and velocity
variables. Since Carleman gave the first proof of global existence \cite{C33}
in the 1930s, many mathematicians have obtained results on the
homogeneous non-relativistic Boltzmann equation, and by now a rather complete
mathematical theory is available. In the present paper we are interested
in generalizations of this to special and general relativity. Compared to the
non-relativistic case, the relativistic equations have not been studied much.
Noutchegueme and his colleagues obtained results on the homogeneous Boltzmann
equation in \cite{ND06,NDT05,NT06,NT03} in several different relativistic
situations, but the scattering kernels they used are not physically
well-motivated. The purpose of this paper is to obtain analogues of their
results in more physically relevant cases.

The scattering kernel is a quantity that determines the nature of collisions
between particles, and in the non-relativistic case several different types of
scattering kernel have been found to be of interest. For instance, the inverse
power law gives the best-known types of scattering kernel, and they are
further classified into hard and soft potential cases. In the relativistic
setting it is not so clear which types of scattering kernel should be of
interest, but a classification of (special) relativistic hard and soft
potentials has been proposed in \cite{DE88,S10} by applying arguments similar
to those used in the non-relativistic case. As in the non-relativistic
case, the scattering kernels depend only on the relative momentum and
scattering angle of two colliding particles. Consider a collision of two
particles, and let $p^\alpha$ and $q^\alpha$ be their momenta before the
collision, and $p'^\alpha$ and $q'^\alpha$ the momenta after the collision.
Then the scattering kernel is given by a function of $p^\alpha-q^\alpha$ and
$p'^\alpha-q'^\alpha$. To be precise, the scattering kernel depends only on the
following two quantities:
\[
(p_\alpha-q_\alpha)(p^\alpha-q^\alpha)\quad\mbox{and}\quad
(p_\alpha-q_\alpha)(p'^\alpha-q'^\alpha),
\]
where the indices are lowered by the Minkowski metric. The above two
quantities are related to the relative momentum and the scattering angle
respectively, and the relativistic hard and soft potentials are defined in
terms of the above quantities. Their precise definitions will be given in
Sections \ref{Sec remarks on the relativistic boltzmann} and
\ref{Sec assumptions of the paper}. On the other hand, the general
relativistic Boltzmann equation does not seem to have been studied enough
concerning specific types of scattering kernel. For instance, in
\cite{B73,BCB73} the authors introduced a quantity $S$ which is a function of
$x$, $p$, $q$, $p'$, and $q'$. The quantity $S$ was used to play
the role of the kernel of the collision operator of the Boltzmann equation,
and strong assumptions on $S$ were made to obtain local well-posedness of the
Einstein-Boltzmann system. However, instead of introducing the abstract
quantity, we may use the principle of general covariance to deal with the
collision operator. If we write down the collision operator with an arbitrary
Lorentzian metric replacing the Minkowski metric, then the special
relativistic collision operator is extended to general relativity in a natural
way. This argument is consistent with the following fact. There exists an
alternative way to write the general relativistic Boltzmann equation. One can
introduce an orthonormal frame $e^\alpha_\mu$ on spacetime and parametrize the
mass shell as $p^\alpha=e^\alpha_\mu v^\mu$. Then, the collision operator no
longer contains any explicit dependence on the metric and reduces to the
special relativistic collision operator. We refer to \cite{LR12} for a more
detailed discussion of this approach.

In this paper, we follow the procedure just mentioned. The general
relativistic collision operator has the same form as in the special
relativistic case, but a certain type of Lorentzian metric will replace the
Minkowski metric. The definitions of hard and soft potentials will be
understood by using the principle of general covariance, and eventually we
will show that the Boltzmann equation has a global solution in the case of a
certain type of hard potential in spatially flat  Robertson-Walker spacetimes.
We refer to \cite{R,W} for basic information about general relativity and to
\cite{A11,dvv,E,St} and Chapter X of \cite{CB} for relativistic kinetic theory.
For classical theories of the non-relativistic Boltzmann equation we refer to
\cite{CIP,G} and their references.

\subsection{Remarks on the relativistic Boltzmann equation}\label{Sec
remarks on the relativistic boltzmann}
Since the Boltzmann equation is an equation describing the dynamics of
collisional matter, it is necessary to understand the collision processes
between particles in order to investigate the equation. The main difference
between the non-relativistic and relativistic collision processes is that
energy and momentum conservation in the non-relativistic case are replaced by
an energy-momentum conservation, and this causes a difficulty in parametrizing
post-collisional momenta in the relativistic case. Let $p^\alpha$ and $q^\alpha$
be two four-vectors describing momenta of two colliding particles in a
relativistic situation, and suppose that the two particles produce momenta
$p'^\alpha$ and $q'^\alpha$ after the collision. Then, energy-momentum
conservation is written as
\[
p^\alpha+q^\alpha=p'^\alpha+q'^\alpha.
\]
Moreover, if all the particles are assumed to have the same mass, then the
mass shell conditions
\[
p_\alpha p^\alpha=q_\alpha q^\alpha=p'_\alpha p'^\alpha=q'_\alpha q'^\alpha=-1
\]
are additionally imposed.
Since the Greek index $\alpha$ runs from $0$ to $3$,
we have six constraints and two free parameters for the post-collisional
momenta, and for the free parameters we use $\omega\in\bbs^2$ as usual.
Consequently, the post-collisional momenta $p'^\alpha$ and $q'^\alpha$ can be
parametrized in terms of the pre-collisional momenta $p^\alpha$ and $q^\alpha$
with an additional parameter $\omega\in\bbs^2$. There are several different
ways known to parametrize post-collisional momenta in the special relativistic
case, see for instance \cite{GS91,S11}. However, the parametrization suggested
in \cite{LR12} will be used in this paper because it can be applied effectively
to the general relativistic case. Suppose that $p^\alpha$ and $q^\alpha$ are
given, and define
\[
n^\alpha:=p^\alpha+q^\alpha\quad\mbox{and}\quad t^\alpha:=(n_i\omega^i,-n_0\omega)
\]
for $\omega\in\bbs^2$. Note that $t^\alpha$ is a general form of vectors
orthogonal to $n^\alpha$, i.e. $n_\alpha t^\alpha=0$ for any $\omega\in\bbs^2$.
The post-collisional momenta are represented as
\begin{equation}\label{parametrization for p' and q'}
p'^\alpha=\frac{p^\alpha+q^\alpha}{2}
+\frac{g}{2}\frac{t^\alpha}{\sqrt{t_\beta t^\beta}}\quad\mbox{and}\quad
q'^\alpha=\frac{p^\alpha+q^\alpha}{2}
-\frac{g}{2}\frac{t^\alpha}{\sqrt{t_\beta t^\beta}}.
\end{equation}
It can easily be checked that they satisfy the mass shell condition and
energy-momentum conservation. The scalar quantity $g$ is called the relative
momentum and is defined by \eqref{s and g} below.

In this paper we are mainly interested in the Robertson-Walker spacetimes but
first we consider the related case of Minkowski space. The spatially
homogeneous Boltzmann equation in Minkowski space is written as follows:
\begin{equation}\label{boltzmann equation in M}
\partial_t f=Q(f,f)
:=\int_{\bbr^3}\int_{\bbs^2} v_\phi\sigma(g,\theta)(f'f'_*-ff_*)
\sqrt{|\eta|}\,d\omega\,dq,\quad v_\phi:=\frac{g\sqrt{s}}{p^0q^0},
\end{equation}
where $f'=f(t,p')$, $f'_*=f(t,q')$, $f=f(t,p)$, and $f_*=f(t,q)$.
Here, $f$ is called the distribution function, $Q$ the collision operator,
$v_\phi$ the M\o{}ller velocity, $\sigma$ the scattering kernel, $\theta$ the
scattering angle, and $\eta$ the determinant of the Minkowski metric.
The first component $p^0$ can be solved for in terms of the other components
$p^i$ due to the mass shell condition as
\[
p^0=\sqrt{1+|p|^2},
\]
where $p=(p^1,p^2,p^3)$ denotes the spatial projection of $p^\alpha$,
and the scalar quantities $s$ and $g$ are given by
\begin{equation}\label{s and g}
s:=-(p_\alpha+q_\alpha)(p^\alpha+q^\alpha),\quad
g:=\sqrt{(p_\alpha-q_\alpha)(p^\alpha-q^\alpha)}.
\end{equation}
They are called the total energy and the relative momentum respectively.
Note that they are conserved quantities in the collision process, i.e.
\[
s=-(p'_\alpha+q'_\alpha)(p'^\alpha+q'^\alpha),\quad
g=\sqrt{(p'_\alpha-q'_\alpha)(p'^\alpha-q'^\alpha)}.
\]
The scattering angle $\theta$ is finally defined as the angle between
$p^\alpha-q^\alpha$ and $p'^\alpha-q'^\alpha$. It can be expressed in terms of
Lorentz invariant quantities as follows:
\begin{equation}\label{scattering angle}
\cos\theta:=\frac{(p_\alpha-q_\alpha)(p'^\alpha-q'^\alpha)}{g^2}
=1-2\frac{(p_\alpha-p'_\alpha)(p^\alpha-p'^\alpha)}{s-4},
\end{equation}
where $s$ and $(p_\alpha-p'_\alpha)(p^\alpha-p'^\alpha)$ are two of the Mandelstam
variables. Note that the scattering angle does not depend on how the
post-collisional momenta are parametrized.

In the second part of the paper we will consider the case of the spatially
flat Robertson-Walker spacetimes. For a given scale factor $R(t)$,
the Boltzmann equation is written as
\[
\partial_t f-2\frac{\dot{R}}{R}\sum_{i=1}^3p^i\partial_{p^i}f=Q(f,f),
\]
where $\dot{R}$ denotes the time derivative of $R$. The collision operator $Q$
is the same as in \eqref{boltzmann equation in M}, but all the quantities in it
such as \eqref{s and g} and \eqref{scattering angle}
are calculated through the Robertson-Walker metric.
For instance, $\sqrt{|\eta|}$ is replaced by $R^3(t)$, and indices are lowered
by the Robertson-Walker metric, so we have
\[
p^0=\sqrt{1+R^2(t)|p|^2}
\]
by the mass shell condition.
However, in this paper we will take a different approach.
It is well-known that in the Robertson-Walker case the Vlasov equation can be
explicitly solved for any given initial data. If the distribution function is
expressed in terms of covariant variables, then it is independent of time,
hence we have $f(t,p_i)=f_0(p_i)$ in the Vlasov case.
Similarly, we can use covariant variables for the Boltzmann equation.
Let us consider again the collision process by using covariant variables.
To make the difference between contravariant and covariant variables clear,
we will use $v_\alpha$ for covariant variables, i.e.
\[
v_\alpha:=p_\alpha\quad\mbox{and}\quad v=(v_1,v_2,v_3).
\]
Then, the energy-momentum conservation is written as
\[
v_\alpha+u_\alpha=v'_\alpha+u'_\alpha,
\]
and the post-collisional momenta $v'_\alpha$ and $u'_\alpha$ are given by
\[
v'_\alpha=\frac{v_\alpha+u_\alpha}{2}
+\frac{g}{2}\frac{t_\alpha}{\sqrt{t_\beta t^\beta}}\quad\mbox{and}\quad
u'_\alpha=\frac{v_\alpha+u_\alpha}{2}
-\frac{g}{2}\frac{t_\alpha}{\sqrt{t_\beta t^\beta}},
\]
where $g$ and $t^\alpha$ are basically the same as in the contravariant case,
but they should be understood as quantities which are constructed from
$v_\alpha$ and $u_\alpha$. Consequently, if the distribution function is
regarded as a function of $t$ and $v$, instead of $p$, then the Boltzmann
equation is written as
\begin{equation}\label{boltzmann equation in RW}
\partial_tf=Q(f,f)=R^{-3}\int_{\bbr^3}\int_{\bbs^2}v_\phi\sigma(g,\theta)
(f'f'_*-ff_*)\,d\omega\,du,\quad v_\phi=\frac{g\sqrt{s}}{v_0u_0},
\end{equation}
where $f'=f(t,v')$, $f'_*=f(t,u')$, $f=f(t,v)$, and $f_*=f(t,u)$.
The first component $v_0$ is solved for in terms of the other components using
the mass shell condition as
\[
v_0=-\sqrt{1+R^{-2}(t)|v|^2},
\]
and the scalar quantities $s$ and $g$ and the scattering angle $\theta$ are
the same as \eqref{s and g} and \eqref{scattering angle} respectively.
The factor $R^{-3}$ comes from the relation
\[
R^3(t)\,dq=R^{-3}(t)\,du.
\]
In the present paper the Minkowski case and the Robertson-Walker case will be
studied separately. In the former case the Boltzmann equation will refer to
\eqref{boltzmann equation in M}, while in the latter case the equation will
refer to \eqref{boltzmann equation in RW}. Comparing the two forms
\eqref{boltzmann equation in M} and \eqref{boltzmann equation in RW},
we might expect that results in the Minkowski case could be easily extended to
the Robertson-Walker case, and this will be done in Section
\ref{Sec hard potential}.

\subsection{Assumptions of the paper}\label{Sec assumptions of the paper}
The scattering kernel $\sigma$ is a function of the relative momentum $g$ and
the scattering angle $\theta$. According to the way that the scattering kernel
depends on its variables, it is classified into hard and soft potentials.
The classification of hard and soft potentials in the relativistic case was
originally introduced by Dudy{\'n}ski and Ekiel-Je{\.z}ewska \cite{DE88} and
recently reformulated by Strain \cite{S10} as follows: for soft potentials we
assume that there exist $\gamma>-2$ and $0< b<\min\{4,\gamma+4\}$ satisfying
\begin{gather*}
\left(\frac{g}{\sqrt{s}}\right)g^{-b}\sigma_0(\omega)
\lesssim\sigma(g,\omega)\lesssim
g^{-b}\sigma_0(\omega),\\
\sigma_0(\omega)\lesssim\sin^\gamma\theta,
\end{gather*}
while for hard potentials we assume that there exist
$\gamma>-2$, $0\leq a\leq\gamma+2$, and $0\leq b<\min\{4,\gamma+4\}$ satisfying
\begin{gather*}
\left(\frac{g}{\sqrt{s}}\right)g^a\sigma_0(\omega)\lesssim\sigma(g,\omega)
\lesssim\left(g^a+g^{-b}\right)\sigma_0(\omega),\\
\sigma_0(\omega)\lesssim\sin^\gamma\theta.
\end{gather*}
Here for any two quantities $A$, $B$ the relation $A\lesssim B$ mens that
there exists a constant $C$ such that $A\le CB$.

\noindent
{\bf The scattering kernel.} In this paper we
assume that the scattering kernel has the form
\begin{equation}\label{scattering kernel}
\sigma(g,\omega)=g^a\sin^\gamma\theta,\quad
-2<\gamma\leq -1,\quad 0\leq a\leq \gamma+2.
\end{equation}
Since $\big(\frac{g}{\sqrt{s}}\big)$ is a bounded quantity
(note that $s=4+g^2$),
a scattering kernel of this form falls into the hard potential case.

\bigskip

Throughout the paper we use weighted $L^1$ spaces in momentum variables. Let
$L^1(\bbr^3)$ be the usual Lebesgue space of integrable functions on $\bbr^3$.
$L^1_r(\bbr^3)$ denotes the weighted $L^1$ space with norm
\begin{equation}
\|f\|_{1,r}:=\int_{\bbr^3}f(y)\langle y\rangle^r\,dy,
\quad \langle y\rangle:=\sqrt{1+|y|^2}.
\end{equation}
The notations which will be used in this paper are as follows:
the Minkowski metric is given by $\mbox{diag}(-1,1,1,1)$,
and the Robertson-Walker metric is
\[
ds^2=-dt^2+R^2(t)((dx^1)^2+(dx^2)^2+(dx^3)^2),
\]
where the scale factor
$R$ is a function of time $t$ satisfying
\[
\dot{R}(t)\geq 0.
\]
For simplicity we assume
\[
R(0)=1.
\]
Greek indices run from $0$ to $3$, and Latin indices
from $1$ to $3$. $p^\alpha$ denotes a four dimensional
vector, while $p$ denotes a three dimensional (contravariant) vector. Similarly,
$v_\alpha$ denotes a four dimensional vector, while $v$ denotes a three
dimensional covariant vector. To be precise,
\[
p=(p^1,p^2,p^3),\quad v=(v_1,v_2,v_3),
\]
while
\[
p^0=\sqrt{1+R^2(t)|p|^2},\quad v^0=-v_0=\sqrt{1+R^{-2}(t)|v|^2}
\]
Note that
\[
p^0=v^0.
\]
The Einstein summation convention is used as
$p_\alpha p^\alpha=\sum_{\alpha=0}^3 p_\alpha p^\alpha$, where the indices are
lowered by $p_\alpha=g_{\alpha\beta}p^\beta$ for a metric $g_{\alpha\beta}$.
In some places we use $p_i p^i=\sum_{i=1}^3 p_i p^i$. The usual inner product
$\cdot$ will only be used for three-dimensional vectors:
\[
p\cdot q=\sum_{i=1}^3p^iq^i,\quad v\cdot u=\sum_{i=1}^3v_iu_i,
\]
and $|\cdot|$ is such that $|p|^2=p\cdot p$ and $|v|^2=v\cdot v$ as usual.

\section{Existence results in the cases of bounded kernels}
\label{Sec bounded kernels}
\setcounter{equation}{0}
In this part we briefly review the results of \cite{ND06,NDT05,NT06,NT03},
where the kernels of the collision operators are assumed to be bounded.
For instance, the homogeneous Boltzmann equation in Minkowski space is written
in \cite{NT03} as
\begin{equation}\label{boltzmann equation with S}
\partial_tf=\frac{1}{p^0}\iint S(p,q,p',q')
(f'f'_*-ff_*)\,d\omega\,\frac{dq}{q^0}.
\end{equation}
An unknown quantity $S$ is introduced to play a role of a kernel of the
collision operator. The collision kernel $S$ is assumed to be bounded
uniformly on pre- and post-collisional momenta with an additional symmetry
assumption:
\begin{gather}
0\leq S\leq C,\label{assumption on S 1}\\
S(p,q,p',q')=S(p',q',p,q),\label{assumption on S 2}
\end{gather}
where $C$ is a fixed constant. The solution obtained in \cite{NT03} should be
understood as a mild solution. By integrating the equation
\eqref{boltzmann equation with S} with respect to the time variable from $0$
to $t$, we obtain the following integral equation:
\begin{equation}\label{boltzmann equation with S mild}
f(t,p)=f_0(p)+\int_0^t\frac{1}{p^0}\iint S(p,q,p',q')
(f'f'_*-ff_*)\,d\omega\,\frac{dq}{q^0}\,ds.
\end{equation}
By saying that $f$ is a solution to the Boltzmann equation
\eqref{boltzmann equation with S}, we mean
$f$ satisfies the integral equation \eqref{boltzmann equation with S mild} in
$L^1$-sense for each $t$. The main theorem can be stated as follows.
\begin{theorem}[Noutchegueme and Tetsadjio, \cite{NT03}]\label{theorem NT03}
Suppose that the collision kernel $S$ satisfies
\eqref{assumption on S 1}--\eqref{assumption on S 2}
for a given positive constant $C$.
Let $r\in(0,\frac{1}{56\pi C}]$ and $f_0\in X_r$ be given for
\[
X_r:=L^1(\bbr^3)\cap\{f\geq 0,\,a.e.\}\cap\{\|f\|_{L^1}\leq r\}.
\]
Then, the Cauchy problem for the homogeneous Boltzmann equation
\eqref{boltzmann equation with S mild}
in Minkowski space has a unique global solution $f\in C([0,\infty);X_r)$
satisfying
\[
\sup_{t\in[0,\infty)}\|f(t)\|_{L^1}\leq \|f_0\|_{L^1}.
\]
\end{theorem}

The above result is extended to the cases of curved spacetimes in
\cite{NDT05,NT06},
and the initial condition is improved to any arbitrary large initial data
(note that the initial condition in the above theorem depends on the constant
$C$,
which is the upper bound of the collision kernel).
The equation considered in \cite{NT06} can be written as
\begin{equation}\label{boltzmann equation with A}
\partial_tf-2\frac{\dot{R}}{R}\sum_{i=1}^3p^i\partial_{p^i}f=\frac{1}{p^0}\iint
A(t,p,q,p',q')(f'f'_*-ff_*)\,d\omega\,\frac{R^3dq}{q^0},
\end{equation}
where the scale factor $R=R(t)$ is assumed to be given such that
it is differentiable and bounded from below,
and the collision kernel $A$ is similarly assumed to be bounded with a
symmetry assumption:
\begin{gather}
R(t)\geq C,\label{assumption on A 1}\\
0\leq A\leq C,\label{assumption on A 2}\\
A(t,p,q,p',q')=A(t,p',q',p,q).\label{assumption on A 3}
\end{gather}
The solution of the equation \eqref{boltzmann equation with A} obtained in
\cite{NT06}
should also be understood as
a mild solution. We define a characteristic curve of
\eqref{boltzmann equation with A} by
\begin{equation}\label{characteristic system}
\frac{dp}{dt}=-2\frac{\dot{R}}{R}p\quad\mbox{with}\quad p(0)=y,\quad
\mbox{hence}\quad p(t)=R^{-2}(t)y.
\end{equation}
Then, the equation \eqref{boltzmann equation with A} can be written as
\[
\frac{d}{dt}f(t,p(t))=(\mbox{collision term}),
\]
and by integrating the above equation in a similar way, we obtain
\[
f(t,p(t))=f_0(y)+\int_0^t(\mbox{collision term})\,ds.
\]
However, if we regard $f$ as a function of $t$ and $y$,
where $y$ and $p=p(t)$ are related to each other by
\eqref{characteristic system}, i.e.
\[
\frac{d}{dt}f(t,y)=(\mbox{collision term}),
\]
then we obtain a different mild form,
\[
f(t,y)=f_0(y)+\int_0^t
(\mbox{collision term})\,ds,
\]
and this was the argument of \cite{NT06}.
To be precise, we express the above integral equation in terms of $y$ and
$z:=R^2(s)q$, and obtain the following mild form of the Boltzmann equation
in the Robertson-Walker spacetime:
\begin{equation}\label{boltzmann equation with A mild}
f(t,y)=f_0(y)+\int_0^t\frac{1}{y^0}\iint
A(s,y,z,y',z')(f'f'_*-ff_*)
\,d\omega\,\frac{dz}{R^3(s)z^0}\,ds.
\end{equation}
The solutions obtained in \cite{NT06} are functions which satisfy
the integral equation \eqref{boltzmann equation with A mild} in $L^1$-sense
for each $t$. If we recall that in this paper the Boltzmann equation in the
Robertson-Walker spacetime refers to \eqref{boltzmann equation in RW} and that
to derive this form of the equation we
considered covariant variables $v$ instead of $p$, then we can see that
considering the characteristic curve corresponds to considering the covariant
variables and that the new variable $y$ of
\eqref{boltzmann equation with A mild}
corresponds to the covariant variable $v$ of \eqref{boltzmann equation in RW}.
The main theorem of \cite{NT06} can be stated as follows.
\begin{theorem}[Noutchegueme and Takou, \cite{NT06}]\label{theorem NT06}
Suppose that the collision kernel $A$ satisfies
\eqref{assumption on A 1}--\eqref{assumption on A 3},
and let $f_0\in L^1(\bbr^3)$ with $f_0\geq 0$ be given. Then, the
Cauchy problem for the
relativistic Boltzmann equation in the Robertson-Walker spacetime
\eqref{boltzmann equation with A mild}
has a unique global solution $f\in C([0,\infty);L^1(\bbr^3))$ with
$f(t)\geq 0$ satisfying
\[
\sup_{t\in[0,\infty)}\|f(t)\|_{L^1}\leq\|f_0\|_{L^1}.
\]
\end{theorem}

In the above result the scale factor $R(t)$ is only assumed to be bounded from
below, hence Theorem \ref{theorem NT06} includes Theorem \ref{theorem NT03}.
A similar result has been obtained in \cite{NDT05}
in a little more general spacetime, which is a Bianchi Type I spacetime,
and finally the Einstein-Boltzmann system has been studied in \cite{ND06}.
The solution space used in \cite{ND06} is a weighted $L^1$ space, which
corresponds to $L^1_1(\bbr^3)$ in the notation of
the present paper. If we combine the arguments and proofs of
\cite{ND06,NDT05,NT06}, the following theorem is obtained.
\begin{proposition}\label{Prop existence for truncated}
Consider the Boltzmann equation \eqref{boltzmann equation with A mild}
in the Robertson-Walker spacetime. Suppose that the collision kernel $A$
satisfies \eqref{assumption on A 1}--\eqref{assumption on A 3},
and let $f_0\in L^1_1(\bbr^3)$ with $f_0\geq 0$ be given.
Then, the Boltzmann equation
has a unique global solution $f\in C([0,\infty);L^1_1(\bbr^3))$
with $f(t)\geq 0$ satisfying
\[
\sup_{t\in[0,\infty)}\|f(t)\|_{L^1_1}\leq\|f_0\|_{L^1_1}.
\]
\end{proposition}

\section{Existence results for the hard potential cases}
\label{Sec hard potential}\setcounter{equation}{0}
In this section the existence result of Proposition
\ref{Prop existence for truncated} will be extended
to the hard potential case.
Proposition \ref{Prop existence for truncated} shows that the Boltzmann
equation \eqref{boltzmann equation with A mild} has a solution in the weighted
function space $L^1_1(\bbr^3)$ with the weight function $\langle y\rangle$. By
assuming that the collision kernel $A$ is independent of $t$
and the scale factor satisfies $R(t)\equiv 1$, we obtain the existence result
in $L^1_1(\bbr^3)$ for the equation \eqref{boltzmann equation with S mild}.
This will be extended to the hard potential case in Section
\ref{Sec hard potential in M}.
The existence result for the Robertson-Walker case will be studied in
Section \ref{Sec hard potential in RW}.

\subsection{Preliminaries}
In this section we collect several lemmas.

\begin{lemma}\label{Lem elementary inequalities}
The following inequalities hold in the Robertson-Walker spacetime:
\begin{gather}
g\leq\sqrt{s},\\
g\leq 2\sqrt{p^0q^0},\\
\sqrt{s}\leq 2\sqrt{p^0q^0},\\
g\leq R(t)|p-q|.
\end{gather}
\end{lemma}
\begin{proof}
Since $s=g^2+4$, we obtain the first inequality.
The second inequality is given by the first and the third ones.
For the third inequality we notice that
\begin{align*}
s&=2+2p^0q^0-2R^2(t)(p\cdot q)\\
&=2p^0q^0+2\sqrt{1-2R^2(t)(p\cdot q)+R^4(t)(p\cdot q)^2}\\
&\leq 2p^0q^0+2\sqrt{1+R^2(t)(|p|^2+|q|^2)+R^4(t)|p|^2|q|^2}\\
&=4p^0q^0,
\end{align*}
and this proves the third inequality.
The last inequality is also clear because
\[
g^2=-(p^0-q^0)^2+R^2(t)|p-q|^2
\leq R^2(t)|p-q|^2,
\]
and this completes the proof.
\end{proof}

\begin{lemma}\label{Lem tt is greater than s}
For any $\omega\in \bbs^2$, we have $t_\alpha t^\alpha\geq R^2(t)s$ in the
Robertson-Walker spacetime.
\end{lemma}
\begin{proof}
In the Robertson-Walker spacetime we have
\[
t^\alpha=(n_i\omega^i,-n_0\omega)=(R^2(t)(n\cdot\omega),n^0\omega),
\]
and the proof is a direct calculation:
\begin{align*}
t_\alpha t^\alpha &=-R^4(t)(n\cdot\omega)^2+R^2(t)(n^0)^2
\geq -R^4(t)|n|^2+R^2(t)(n^0)^2\\
&=R^2(t)((n^0)^2-R^2(t)|n|^2)=R^2(t)s.
\end{align*}
This completes the proof.
\end{proof}

The following lemma is a well-known fact in the special relativistic case
\cite{GS91}, and we show that the lemma holds in general relativistic cases
also. The lemma is proved by the same argument as in \cite{GS91}, but we
present it for the reader's convenience.

\begin{lemma}\label{Lem change of variables}
Let $(p'^\alpha,q'^\alpha)$ and $(p^\alpha,q^\alpha)$ be pre- and post-collisional
momenta respectively, and consider the collision map
$(p^\alpha,q^\alpha)\to(p'^\alpha,q'^\alpha)$. Then, the Jacobian is given by
\[
\frac{\partial(p',q')}{\partial(p,q)}=-\frac{p'_0q'_0}{p_0q_0}.
\]
\end{lemma}
\begin{proof}
To prove the lemma, we use a different parametrization from
\eqref{parametrization for p' and q'}.
\[
p'^\alpha=p^\alpha+2\frac{t_\beta q^\beta}{t_\gamma t^\gamma}t^\alpha,\quad
q'^\alpha=q^\alpha-2\frac{t_\beta q^\beta}{t_\gamma t^\gamma}t^\alpha,
\]
where $t^\alpha$ is the same as in \eqref{parametrization for p' and q'}.
For convenience we write
\[
p'^k=p^k+A\omega^k,\quad q'^k=q^k-A\omega^k,\quad A
=-2\frac{t_\beta q^\beta}{t_\gamma t^\gamma}n_0.
\]
By the same calculations as in \cite{GS91} we obtain
\begin{align}
\frac{\partial(p',q')}{\partial(p,q)}&=\det\left(
\begin{array}{cc}
\delta_j^i+(\partial_{p^j}A)\omega^i & (\partial_{q^j}A)\omega^i\\
-(\partial_{p^j}A)\omega^i & \delta_j^i-(\partial_{q^j}A)\omega^i
\end{array}
\right)\cr
&=\det\left(
\delta_j^i+(\partial_{p^j}A-\partial_{q^j}A)\omega^i
\right)\cr
&=1+(\partial_{p^i}A-\partial_{q^i}A)\omega^i.\label{determinant 1}
\end{align}
We differentiate the conserved energy
\[
p'^0+q'^0=p^0+q^0
\]
with respect to $p^j$, and multiply $\omega^j$ to obtain
\begin{align*}
&\left(-\frac{p'_k}{p'_0}\frac{\partial p'^k}{\partial p^j}
-\frac{q'_k}{q'_0}\frac{\partial q'^k}{\partial p^j}\right)\omega^j
=-\frac{p_j}{p_0}\omega^j\\
&\Longleftrightarrow\quad
\left(\frac{p_j}{p_0}-\frac{p'_j}{p'_0}\right)\omega^j
=\left(\frac{p'_k}{p'_0}-\frac{q'_k}{q'_0}\right)
\omega^k(\partial_{p^j}A)\omega^j.
\end{align*}
Similarly we obtain
\[
\left(\frac{q_j}{q_0}-\frac{q'_j}{q'_0}\right)\omega^j
=\left(\frac{p'_k}{p'_0}-\frac{q'_k}{q'_0}\right)
\omega^k(\partial_{q^j}A)\omega^j.
\]
Hence, \eqref{determinant 1} is given by
\begin{equation}\label{determinant 2}
\frac{\partial(p',q')}{\partial(p,q)}
=\left(\frac{p'_k}{p'_0}\omega^k-\frac{q'_k}{q'_0}\omega^k\right)^{-1}
\left(\frac{p_j}{p_0}\omega^j-\frac{q_j}{q_0}\omega^j\right).
\end{equation}
Recall that $n^\alpha=p^\alpha+q^\alpha$ and $t^\alpha=(n_j\omega^j,-n_0\omega)$
for $\omega\in\bbs^2$, and then the above quantities are written as follows:
\begin{align}
\frac{p_j}{p_0}\omega^j-\frac{q_j}{q_0}\omega^j
=\frac{1}{p_0q_0}(q_0n_j-n_0q_j)\omega^j
=\frac{1}{p_0q_0}q_\alpha t^\alpha.\label{determinant 3}
\end{align}
Similarly we obtain
\begin{align}
\frac{p'_k}{p'_0}\omega^k-\frac{q'_k}{q'_0}\omega^k
&=\frac{1}{p'_0q'_0}(q'_0p'_k+q'_0q'_k-q'_0q'_k-p'_0q'_k)\omega^k\cr
&=\frac{1}{p'_0q'_0}(q'_0n_k-n_0q'_k)\omega^k
=\frac{1}{p'_0q'_0}q'_\alpha t^\alpha
=-\frac{1}{p'_0q'_0}q_\alpha t^\alpha,\label{determinant 4}
\end{align}
where we used the energy-momentum conservation and the following:
\[
t_\alpha q'^\alpha=t_\alpha q^\alpha
-2\frac{t_\beta q^\beta}{t_\gamma t^\gamma}t_\alpha t^\alpha
=-t_\alpha q^\alpha.
\]
We plug \eqref{determinant 3} and \eqref{determinant 4} into
\eqref{determinant 2}, and this completes the proof.
\end{proof}

\begin{lemma}\label{Lem f'g'-fg to p'+q'-p-q}
For the collision operator the following property holds in the
Robertson-Walker spacetime: for any measurable function $k$ depending only on
$g$, $s$, and $\omega$, we have
\begin{align*}
&\iiint \frac{k(g,s,\omega)}{p^0q^0}(f'f_*'-ff_*)(p^0)^r\,d\omega\,dq\,dp\\
&=\frac{1}{2}\iiint \frac{k(g,s,\omega)}{p^0q^0}ff_*((p'^0)^r+(q'^0)^r-(p^0)^r
-(q^0)^r)\,d\omega\,dq\,dp.
\end{align*}
\end{lemma}
\begin{proof}
We use Lemma \ref{Lem change of variables} to make the change of variables
between pre- and post-collisional momenta as follows:
\begin{equation}
\frac{1}{p^0q^0}dp\,dq=\frac{1}{p'^0q'^0}dp'dq',
\end{equation}
and note that $g$ and $s$ are invariant quantities
under the collision process and symmetric for $p$ and $q$.
Hence, the gain term can be written as
\begin{align*}
\iiint \frac{k(g,s,\omega)}{p^0q^0}f'f_*'(p^0)^r\,d\omega\,dq\,dp	
&=\iiint\frac{k(g,s,\omega)}{p'^0q'^0}f'f_*'(p^0)^r\,d\omega\,dq'dp'\\
&=\iiint\frac{k(g,s,\omega)}{p^0q^0}ff_*(p'^0)^r\,d\omega\,dq\,dp.
\end{align*}
By interchanging $p$ and $q$, it can also be rewritten as
\begin{align*}
\iiint \frac{k(g,s,\omega)}{p^0q^0}f'f_*'(p^0)^r\,d\omega\,dq\,dp	
&=\iiint\frac{k(g,s,\omega)}{p^0q^0}ff_*(q'^0)^r\,d\omega\,dq\,dp.
\end{align*}
Hence, we obtain the following representation for the gain term:
\begin{align*}
\iiint \frac{k(g,s,\omega)}{p^0q^0}f'f_*'(p^0)^r\,d\omega\,dq\,dp	
&=\frac{1}{2}\iiint\frac{k(g,s,\omega)}{p^0q^0}ff_*((p'^0)^r+(q'^0)^r)
\,d\omega\,dq\,dp.
\end{align*}
After applying the same argument to the loss term, we obtain the
desired result.
\end{proof}

\begin{lemma}\label{Lem estimate of G}
Consider the collision process in the Robertson-Walker spacetime.
Let $(p'^\alpha,q'^\alpha)$ and $(p^\alpha,q^\alpha)$ be pre- and post-collisional
momenta respectively. Consider the following quantity for $r>1$:
\[
G=(p'^0)^r+(q'^0)^r-(p^0)^r-(q^0)^r.
\]
Then, $G$ satisfies
\begin{equation}
G\leq C_r((p^0)^{r-1}q^0+p^0(q^0)^{r-1}).
\end{equation}
If $\omega$ is restricted to a subset
$\{\omega\in \bbs^2:|n\cdot\omega|\leq \frac{1}{\sqrt{2}}|n|\}$,
then $G$ satisfies
\begin{equation}
G\leq C_r((p^0)^{r-\frac{1}{2}}(q^0)^{\frac{1}{2}}
+(p^0)^{\frac{1}{2}}(q^0)^{r-\frac{1}{2}})-c_r((p^0)^r+(q^0)^r),
\end{equation}
where $C_r$ and $c_r$ are two different positive constants depending on $r$.
\end{lemma}
\begin{proof}
Note that $p'^0+q'^0=p^0+q^0$ is a conserved quantity
for each $p^0$ and $q^0$. Let $p^\alpha$ and $q^\alpha$ be given.
Then, $G$ reduces to a function of $\omega$ and can be written as
\begin{align*}
G(\omega)&=(p'^0)^r+(q'^0)^r-(p^0)^r-(q^0)^r\cr
&=(p'^0)^r+(p^0+q^0-p'^0)^r-(p^0)^r-(q^0)^r.
\end{align*}
Note that $G$ has minimum at $p'^0=(p^0+q^0)/2$, i.e., when $p'^0=q'^0$, and
is monotonically increasing as $p'^0$ tends to $0$ or $p^0+q^0$.
Hence, $G$ attains its maximum when $p'^0-q'^0$ is extremal.
Without loss of generality we may assume $p'^0\geq q'^0$.
From \eqref{parametrization for p' and q'} we obtain
\begin{align*}
p'^0-q'^0&=g\frac{t^0}{\sqrt{t_\alpha t^\alpha}}
=g\frac{R^2(n\cdot\omega)}{\sqrt{R^2(n^0)^2-R^4(n\cdot\omega)^2}}.
\end{align*}
This quantity attains its maximum when $\omega$ is parallel to $n$,
which means that $p'^0$ is largest and $q'^0$ is smallest when
$n\cdot\omega=|n|$, in particular $q'^0\leq \min\{p^0,q^0\}$, which implies
again $(q'^0)^r\leq \min\{(p^0)^r,(q^0)^r\}$. Consequently, $G$ attains its
maximum when $n\cdot\omega=|n|$ and is estimated as
\begin{align*}
G(\omega)&\leq G(n/|n|)\leq \left.(p'^0)^r\right|_{\omega=n/|n|}
-\max\{(p^0)^r,(q^0)^r\}\cr
&=\left.\left(\frac{p^0+q^0}{2}
+\frac{g}{2}\frac{R^2|n|}{\sqrt{t_\alpha t^\alpha}}\right)^r\right|_{\omega=n/|n|}
-\max\{(p^0)^r,(q^0)^r\}.
\end{align*}
By applying Lemma \ref{Lem elementary inequalities}, Lemma
\ref{Lem tt is greater than s}, and the inequality
\begin{equation}\label{estimate of a plus b ^r}
(a+b)^r\leq a^r+b^r+ C_r(a^{r-1}b+ab^{r-1})\quad\mbox{for}\quad r>1,
\end{equation}
we obtain
\begin{align*}
G(\omega)&\leq \left(\frac{p^0+q^0}{2}+\frac{R|p+q|}{2}\right)^r
-\max\{(p^0)^r,(q^0)^r\}\cr
&\leq (p^0+q^0)^r-\max\{(p^0)^r,(q^0)^r\}\cr
&\leq \min\{(p^0)^r,(q^0)^r\}+C_r((p^0)^{r-1}q^0+p^0(q^0)^{r-1})\cr
&\leq C_r((p^0)^{r-1}q^0+p^0(q^0)^{r-1}),
\end{align*}
and this proves the first result.

To prove the second result, we take the assumption
$|n\cdot\omega|\leq \frac{1}{\sqrt{2}}|n|$ and suppose $p'^0\geq q'^0$. Then,
$p'^0$ is estimated as
\begin{align*}
p'^0&\leq \frac{p^0+q^0}{2}+\frac{g}{2}
\frac{R^2|n\cdot\omega|}{\sqrt{R^2(n^0)^2-R^4(n\cdot\omega)^2}}\\
&\leq \frac{p^0+q^0}{2}+\frac{g}{2\sqrt{2}}\frac{R^2|p+q|}{\sqrt{R^2(p^0+q^0)^2
-\frac{1}{2}R^4|p+q|^2}}
\leq \frac{p^0+q^0}{2}+\frac{g}{2}\\
&\leq \frac{(\sqrt{p^0}+\sqrt{q^0})^2}{2},
\end{align*}
where we used Lemma \ref{Lem elementary inequalities}.
Then, $G$ is estimated as
\begin{align*}
G&\leq 2(p'^0)^r-(p^0)^r-(q^0)^r\cr
&\leq \frac{(\sqrt{p^0}+\sqrt{q^0})^{2r}}{2^{r-1}}-(p^0)^r-(q^0)^r\cr
&\leq \frac{(p^0)^r}{2^{r-1}}+\frac{(q^0)^r}{2^{r-1}}
+C_r((p^0)^{r-\frac{1}{2}}(q^0)^{\frac{1}{2}}
+(p^0)^{\frac{1}{2}}(q^0)^{r-\frac{1}{2}})-(p^0)^r-(q^0)^r\cr
&\leq C_r((p^0)^{r-\frac{1}{2}}(q^0)^{\frac{1}{2}}
+(p^0)^{\frac{1}{2}}(q^0)^{r-\frac{1}{2}})-c_r((p^0)^r+(q^0)^r),
\end{align*}
where we used \eqref{estimate of a plus b ^r}.
The constants $C_r$ and $c_r$ are two different positive constants depending
on $r$, and this completes the proof.
\end{proof}
\begin{remark}
Those types of inequalities given in Lemma \ref{Lem estimate of G} are called
the Povzner inequality, which was originally proved by Povzner in \cite{P62}.
This inequality has been crucially used to prove existence theorems for the
non-relativistic spatially homogeneous Boltzmann equation by Elmroth \cite{E83}
and Mischler and Wennberg \cite{MW99}. The sharpest form of the Povzner
inequality is given by Mischler and Wennberg, but Lemma
\ref{Lem estimate of G} corresponds to a relativistic extension of Elmroth's
result.
\end{remark}

\subsection{Hard potential case in Minkowski space}
\label{Sec hard potential in M}
Comparing the two forms \eqref{boltzmann equation in M} and
\eqref{boltzmann equation with S} for the Boltzmann equation, we can see that
the collision kernel $S$ of \eqref{boltzmann equation with S}
corresponds to $g\sqrt{s}\sigma(g,\theta)$ of \eqref{boltzmann equation in M}.
Hence, if the quantity $g\sqrt{s}\sigma(g,\theta)$ is suitably truncated,
then the truncated equation has a global solution
by Proposition \ref{Prop existence for truncated}.
For simplicity the following notations will be used:
\[
v_{\phi,m}:=\frac{\min\{g\sqrt{s},m\}}{p^0q^0},\quad
g_m:=\min\{g,m\},\quad
\sigma_{0,m}(\omega):=\min\{\sigma_0(\omega),m\}.
\]
For each integer $m$, let $f_m$ be a solution of the following truncated
equation with initial data $f_m(0)=f_0$:
\begin{equation}\label{boltzmann equation truncated}
\partial_t f_m=Q_m(f_m,f_m),
\end{equation}
where $Q_m$ is defined as
\[
Q_m(h,h):=\iint v_{\phi,m}(g_m)^a\sigma_{0,m}
(\omega)(h'h'_{*}-hh_{*})\,d\omega\,dq.
\]
Then, the truncated equation has a unique global solution
$f_m\in C([0,\infty);L^1_1(\bbr^3))$ by Proposition
\ref{Prop existence for truncated}.
Consequently, a sequence $\{f_m\}$ is obtained,
and existence of solutions for the original equation
\eqref{boltzmann equation in M} is proved by showing that the sequence
$\{f_m\}$ is a Cauchy sequence in $L^1_1(\bbr^3)$. This argument was used for
the non-relativistic case by Mischler and Wennberg in \cite{MW99}, and below we
will show that their argument is applicable to the relativistic case.

\begin{lemma}\label{Lem uniform boundedness of Y_r}
For any $r\geq 0$ and $T>0$, there exists a constant $C_r$ which does not
depend on $m$ such that if $\|f_0\|_{1,r}$ is bounded, then
\[
\sup_m\sup_{t\in[0,T]}\|f_m(t)\|_{1,r}\leq C_r.
\]
\end{lemma}
\begin{proof}
We first note that by Proposition \ref{Prop existence for truncated}
\[
\sup_{t\in[0,\infty)}\|f_m(t)\|_{1,r}\leq C\quad\mbox{for}\quad 0\leq r\leq 1,
\]
where $C$ does not depend on $m$, and for $r\leq s$,
\[
\|f_m(t)\|_{1,r}\leq \|f_m(t)\|_{1,s}.
\]
We now assume $r>1$.
By direct calculations we have
\begin{align*}
&\frac{d}{dt}\|f_m(t)\|_{1,r}\\
&=\iiint v_{\phi,m}(g_m)^a\sigma_{0,m}(\omega)
(f_m'f_{m*}'-f_mf_{m*})(p^0)^r\,d\omega\,dq\,dp\\
&=\frac{1}{2}\iiint v_{\phi,m} (g_m)^a\sigma_{0,m}(\omega)f_mf_{m*}
((p'^0)^r+(q'^0)^r-(p^0)^r-(q^0)^r)\,d\omega\,dq\,dp,
\end{align*}
where we used Lemma \ref{Lem f'g'-fg to p'+q'-p-q}.
We apply Lemma \ref{Lem estimate of G} to obtain
\[
\frac{d}{dt}\|f_m(t)\|_{1,r}\leq I_1+I_2-I_3,
\]
where
\begin{align*}
I_1&= C_r\iiint_{|n\cdot\omega|\geq\frac{1}{\sqrt{2}}|n|}
v_{\phi,m}(g_m)^a\sigma_{0,m}(\omega) f_mf_{m*}
((p^0)^{r-1}q^0+p^0(q^0)^{r-1})\,d\omega\,dq\,dp,\\
I_2&= C_r\iiint_{|n\cdot\omega|\leq\frac{1}{\sqrt{2}}|n|}
v_{\phi,m}(g_m)^a\sigma_{0,m}(\omega) f_mf_{m*}
((p^0)^{r-\frac{1}{2}}(q^0)^{\frac{1}{2}}
+(p^0)^\frac{1}{2}(q^0)^{r-\frac{1}{2}})\,d\omega\,dq\,dp,\\
I_3&= c_r\iiint_{|n\cdot\omega|\leq\frac{1}{\sqrt{2}}|n|}
v_{\phi,m}(g_m)^a\sigma_{0,m}(\omega) f_mf_{m*}
((p^0)^r+(q^0)^r)\,d\omega\,dq\,dp.
\end{align*}
The second term $I_2$ is easily estimated by using Lemma
\ref{Lem elementary inequalities} as
\begin{align*}
I_2&\leq C_r\iint f_mf_{m*}(p^0)^{r-\frac{1}{2}+\frac{a}{2}}
(q^0)^{\frac{1}{2}+\frac{a}{2}}\,dq\,dp
\leq C_r \|f_m(t)\|_{1,r-\frac{1}{2}+\frac{a}{2}}\|f_m(t)\|_{1,\frac{1}{2}+\frac{a}{2}}.
\end{align*}

Consider now $\sigma_{0,m}(\omega)$, which is defined by
\[
\sigma_{0,m}(\omega):=\min\{\sin^\gamma\theta,m\}\quad\mbox{for}\quad
-2<\gamma\leq -1.
\]
Note that $\sigma_{0,m}(\omega)$ is integrable on $\bbs^2$ for $\gamma>-2$,
and there exists a constant $C_\gamma$ satisfying
$\int_{\bbs^2}\sigma_{0,m}(\omega)\,d\omega\leq C_\gamma$,
where the constant $C_\gamma$ does not depend on $m$. On the other hand, since
$\gamma$ is negative, we have $\sigma_{0,m}(\omega)\geq 1$ for any $m$.
Moreover, the integration domain of $I_3$ is a set with Lebesgue measure
\[
\mu\{|n\cdot\omega|\leq |n|/\sqrt{2}\}=2\sqrt{2}\pi,
\]
which does not depend on $m$.
Hence, $I_1$ and $I_3$ can be estimated as
\begin{align*}
I_1&\leq D_r\iint_{\bbr^6} v_{\phi,m}(g_m)^a f_mf_{m*}(p^0)^{r-1}q^0\,dq\,dp,\\
I_3&\geq d_r\iint_{\bbr^6} v_{\phi,m}(g_m)^a f_mf_{m*}(q^0)^r\,dq\,dp,
\end{align*}
for some constants $D_r$ and $d_r$.
We now fix the constants $D_r$ and $d_r$ to split the domain
by $\{D_r(p^0)^{r-1}\leq d_r(q^0)^{r-1}\}$ and
$\{D_r(p^0)^{r-1}\geq d_r(q^0)^{r-1}\}$,
and then obtain $I_1\leq I_{11}+I_{12}$, where
\begin{align*}
I_{11}&= D_r\iint_{D_r(p^0)^{r-1}\leq d_r(q^0)^{r-1}}
v_{\phi,m}(g_m)^a f_mf_{m*}(p^0)^{r-1}q^0\,dq\,dp,\cr
I_{12}&= D_r\iint_{D_r(p^0)^{r-1}\geq d_r(q^0)^{r-1}}
v_{\phi,m}(g_m)^a f_mf_{m*}(p^0)^{r-1}q^0\,dq\,dp.
\end{align*}
We now obtain
\[
I_{11}\leq I_3.
\]
In the case of $I_{12}$, we may simply use
$(g_m)^a\leq C(p^0q^0)^{\frac{a}{2}}\leq C_r(p^0)^a$.
Then, $I_{12}$ is easily estimated as
\[
I_{12}\leq C_r \|f_m(t)\|_{1,r-1+a}.
\]
Combining the above estimates, we obtain
\begin{align*}
&\frac{d}{dt}\|f_m(t)\|_{1,r}\\
&\leq C_r\left( \|f_m(t)\|_{1,r-\frac{1}{2}+\frac{a}{2}}
\|f_m(t)\|_{1,\frac{1}{2}+\frac{a}{2}}+\|f_m(t)\|_{1,r-1+a}\right)
\leq C_r\|f_m(t)\|_{1,r},
\end{align*}
where we used the fact that $0\leq a\leq\gamma+2$ and $-2<\gamma\leq -1$.
Then, the lemma is proved by applying Gr{\"o}nwall's inequality.
\end{proof}
\begin{lemma}\label{Lem Cauchy sequence}
Consider the sequence $\{f_m\}$ on any finite time interval $[0,T]$.
For any small number $\delta>0$,
there exists a positive integer $M$ such that if $k,m\geq M$, then
\[
\sup_{t\in[0,T]}\|f_k(t)-f_m(t)\|_{1,1}\leq \delta.
\]
\end{lemma}
\begin{proof}
Let $k\leq m$ be two positive integers. By direct calculations,
\[
\frac{d}{dt}\|f_k(t)-f_m(t)\|_{1,1}= I+J,
\]
where
\begin{align*}
I&=\int\mbox{sgn}(f_k-f_m)(Q_k(f_k,f_k)-Q_k(f_m,f_m))p^0\,dp,\\
J&=\int\mbox{sgn}(f_k-f_m)(Q_k(f_m,f_m)-Q_m(f_m,f_m))p^0\,dp,
\end{align*}
and $I$ and $J$ will be estimated separately.
The first term $I$ is split again as
\begin{align*}
I&=\frac{1}{2}\iiint\mbox{sgn}(f_k-f_m)v_{\phi,k}(g_k)^a\sigma_{0,k}(\omega)\\
&\quad\times((f_k'-f_m')(f_{k*}'+f_{m*}')+(f_k'+f_m')(f_{k*}'-f_{m*}')\\
&\quad\quad-(f_k-f_m)(f_{k*}+f_{m*})-(f_k+f_m)(f_{k*}-f_{m*}))
p^0\,d\omega\,dq\,dp\\
&=:I_1+I_2+I_3+I_4.
\end{align*}
Each $I_i$ is estimated as follows:
\begin{align*}
I_1&\leq\frac{1}{2}\iiint v_{\phi,k}(g_k)^a\sigma_{0,k}(\omega)
|f_k'-f_m'|(f_{k*}'+f_{m*}')p^0\,d\omega\,dq\,dp\\
&=\frac{1}{2}\iiint v_{\phi,k}(g_k)^a\sigma_{0,k}(\omega)
|f_{k}-f_{m}|(f_{k*}+f_{m*})p'^0\,d\omega\,dq\,dp,
\end{align*}
\begin{align*}
I_2&\leq\frac{1}{2}\iiint v_{\phi,k}(g_k)^a\sigma_{0,k}(\omega)
(f_k'+f_m')|f_{k*}'-f_{m*}'|p^0\,d\omega\,dq\,dp\\
&=\frac{1}{2}\iiint v_{\phi,k}(g_k)^a\sigma_{0,k}(\omega)
(f_{k*}+f_{m*})|f_{k}-f_{m}|q'^0\,d\omega\,dq\,dp,
\end{align*}
\begin{align*}
I_3&=-\frac{1}{2}\iiint\mbox{sgn}(f_k-f_m)v_{\phi,k}(g_k)^a\sigma_{0,k}(\omega)
(f_k-f_m)(f_{k*}+f_{m*})p^0\,d\omega\,dq\,dp\cr
&=-\frac{1}{2}\iiint v_{\phi,k}(g_k)^a\sigma_{0,k}(\omega)
|f_{k}-f_{m}|(f_{k*}+f_{m*})p^0\,d\omega\,dq\,dp,
\end{align*}
and finally
\begin{align*}
I_4&\leq \frac{1}{2}\iiint v_{\phi,k}(g_k)^a\sigma_{0,k}(\omega)
(f_k+f_m)|f_{k*}-f_{m*}|p^0\,d\omega\,dq\,dp\\
&\leq\frac{1}{2}\iiint v_{\phi,k}(g_k)^a\sigma_{0,k}(\omega)
(f_{k*}+f_{m*})|f_{k}-f_{m}|q^0\,d\omega\,dq\,dp.
\end{align*}
Therefore, $I$ is estimated as
\begin{align*}
I&\leq \frac{1}{2}\iiint v_{\phi,k} (g_k)^a\sigma_{0,k}(\omega)
|f_{k}-f_{m}|(f_{k*}+f_{m*})
(p'^0+q'^0-p^0+q^0)\,d\omega\,dq\,dp\cr
&\leq C\iint (g_k)^a|f_{k}-f_{m}|(f_{k*}+f_{m*})q^0\,dq\,dp,
\end{align*}
where we used
\[
p'^0+q'^0=p^0+q^0.
\]
By using $g_k\leq 2\sqrt{p^0q^0}$,
we obtain for $I$
\begin{align}
I&\leq C\iint |f_{k}-f_{m}|(f_{k*}+f_{m*})(p^0)^{\frac{a}{2}}
(q^0)^{1+\frac{a}{2}}\,dq\,dp\cr
&\leq C\sup_n\|f_n(t)\|_{1,1+\frac{a}{2}}\|f_k(t)-f_m(t)\|_{1,\frac{a}{2}}.
\label{estimate of I}
\end{align}

To estimate the second term $J$, we note that
\begin{align*}
|v_{\phi,k}-v_{\phi,m}|&=\frac{1}{p^0q^0}|\min\{g\sqrt{s},k\}
-\min\{g\sqrt{s},m\}|\\
&\leq {\bf 1}_{\{g\sqrt{s}\geq k\}}\frac{\min\{g\sqrt{s},m\}}{p^0q^0}
={\bf 1}_{\{g\sqrt{s}\geq k\}}v_{\phi,m},
\end{align*}
and similarly
\begin{align*}
&|(g_k)^a-(g_m)^a|\leq {\bf 1}_{\{g\geq k\}}(g_m)^a,\\
&|\sigma_{0,k}(\omega)-\sigma_{0,m}(\omega)|\leq {\bf 1}_{\{\sin^\gamma\theta\geq k\}}
\sigma_{0,m}(\omega).
\end{align*}
Hence, $J$ can be estimated as
\begin{align*}
J&\leq\iiint |v_{\phi,k}(g_k)^a\sigma_{0,k}(\omega)
-v_{\phi,m}(g_m)^a\sigma_{0,m}(\omega)|
|f_m'f'_{m*}-f_mf_{m*}|p^0\,d\omega\,dq\,dp\\
&\leq \iiint {\bf 1}_{\{g\sqrt{s}\geq k\}}v_{\phi,m}(g_k)^a\sigma_{0,k}(\omega)
(f_m'f'_{m*}+f_mf_{m*})p^0\,d\omega\,dq\,dp\\
&\quad +\iiint {\bf 1}_{\{g\geq k\}}v_{\phi,m}(g_m)^a\sigma_{0,k}(\omega)
(f_m'f'_{m*}+f_mf_{m*})p^0\,d\omega\,dq\,dp\\
&\quad +\iiint {\bf 1}_{\{\sin^\gamma\theta\geq k\}}v_{\phi,m}(g_m)^a\sigma_{0,m}(\omega)
(f_m'f'_{m*}+f_mf_{m*})p^0\,d\omega\,dq\,dp\\
&=:J_1+J_2+J_3.
\end{align*}
Note that each $J_i$ can be separated into two terms: a gain term containing
$f'_mf'_{m*}$ and a loss term containing $f_mf_{m*}$. The gain and loss terms
are estimated in the same way after making the change of variables
$(p,q)\leftrightarrow(p',q')$, hence we only present the estimates for the
loss terms. To estimate $J_1$, we take a small number $\varepsilon>0$ and
use $g\sqrt{s}\leq 4p^0q^0$ from Lemma \ref{Lem elementary inequalities}:
\begin{align}
J_1&\leq C\iint{\bf 1}_{\{4p^0q^0\geq k\}}(g_k)^af_mf_{m*}p^0\,dq\,dp\cr
&\leq C\iint {\bf 1}_{\{4p^0q^0\geq k\}}f_m(p^0)^{1+\frac{a}{2}}
f_{m*}(q^0)^{\frac{a}{2}}\,dq\,dp\cr
&\leq \frac{C}{k^\varepsilon}
\iint {\bf 1}_{\{4p^0q^0\geq k\}}f_m(p^0)^{1+\frac{a}{2}+\varepsilon}
f_{m*}(q^0)^{\frac{a}{2}+\varepsilon}\,dq\,dp\cr
&\leq \frac{C}{k^\varepsilon}\|f_m(t)\|_{1,1+\frac{a}{2}+\varepsilon}
\|f_m(t)\|_{1,\frac{a}{2}+\varepsilon}.
\label{estimate of J1}
\end{align}
To estimate $J_2$, we use $g\leq|p-q|$ to obtain
\begin{align}
J_2&\leq C\iint {\bf 1}_{\{|p-q|\geq k\}}f_m(p^0)^{1+\frac{a}{2}}f_{m*}
(q^0)^\frac{a}{2}\,dq\,dp\cr
&\leq C\iint{\bf 1}_{\{|p|\geq \frac{k}{2}\}\cup\{|q|\geq \frac{k}{2}\}}
f_m(p^0)^{1+\frac{a}{2}}f_{m*}(q^0)^{\frac{a}{2}}\,dq\,dp\cr
&\leq C\|f_m(t)\|_{1,1+\frac{a}{2}}\int{\bf 1}_{\{|q|\geq\frac{k}{2}\}}f_{m*}
(q^0)^{\frac{a}{2}}\,dq\cr
&\leq \frac{C}{k}\|f_m(t)\|_{1,1+\frac{a}{2}}^2.
\label{estimate of J2}
\end{align}
For $J_3$ term, we use $\sin\theta\approx\theta$ for
$0\leq\theta\leq\frac{\pi}{2}$.
Hence, the condition $\sin^\gamma\theta\geq k$ is equivalent to
$\theta\leq Ck^{\frac{1}{\gamma}}$ since $\gamma$ is negative.
We first estimate $J_3$ as
\begin{align*}
J_3&\leq C\iiint {\bf 1}_{\{\theta\leq Ck^\frac{1}{\gamma}\}}\sigma_{0,m}(\omega)
f_m(p^0)^{1+\frac{a}{2}}f_{m*}(q^0)^{\frac{a}{2}}\,d\omega\,dq\,dp\cr
&\leq C\|f_m(t)\|_{1,1+\frac{a}{2}}\|f_m(t)\|_{1,\frac{a}{2}}
\int{\bf 1}_{\{\theta\leq Ck^\frac{1}{\gamma}\}}\sigma_{0,m}(\omega)\,d\omega.
\end{align*}
The integration on $\bbs^2$ above is estimated as
\begin{align*}
\int{\bf 1}_{\{\theta\leq Ck^\frac{1}{\gamma}\}}\sigma_{0,m}(\omega)\,d\omega
&\leq 2\pi\int_0^{Ck^{\frac{1}{\gamma}}}\sin^{\gamma+1}\theta\,d\theta
\leq Ck^{\frac{\gamma+2}{\gamma}},
\end{align*}
where the constant depends on $\gamma$. Note that
$-1\leq (\gamma+2)/\gamma<0$, and the third term $J_3$ is estimated as
\begin{equation}\label{estimate of J3}
J_3\leq Ck^{\frac{\gamma+2}{\gamma}}\|f_m(t)\|_{1,1+\frac{a}{2}}\|f_m(t)\|_{1,\frac{a}{2}}.
\end{equation}
We combine \eqref{estimate of I}, \eqref{estimate of J1},
\eqref{estimate of J2},
and \eqref{estimate of J3}, and apply Lemma \ref{Lem uniform boundedness of Y_r}
on any finite time interval $[0,T]$ to obtain
\begin{align*}
\frac{d}{dt}\|f_k(t)-f_m(t)\|_{1,1}&\leq C(k^{-\varepsilon}
+k^{-1}+k^{\frac{\gamma+2}{\gamma}})
+C\|f_k(t)-f_m(t)\|_{1,\frac{a}{2}}\\
&\leq C(k^{-\varepsilon}+k^{\frac{\gamma+2}{\gamma}})+C\|f_k(t)-f_m(t)\|_{1,1}.
\end{align*}
Since $f_k(0)=f_m(0)$ and $(\gamma+2)/\gamma$ is negative,
we obtain the desired result by applying Gr{\"o}nwall's inequality.
\end{proof}
We now obtain the following theorem.
\begin{theorem}\label{Thm main 1}
Suppose that the scattering kernel has the form  \eqref{scattering kernel}
and initial data satisfies $f_0\in L^1_{r}(\bbr^3)$
for $r>1+\frac{a}{2}$ with $f_0\geq 0$.
Then, the Boltzmann equation \eqref{boltzmann equation in M} in Minkowski space
has a unique global solution $f\in C([0,\infty);L^1_1(\bbr^3))$ with $f(t)\geq 0$.
\end{theorem}
\begin{proof}
Lemma \ref{Lem Cauchy sequence} shows that the sequence $\{f_m\}$ is a Cauchy
sequence in $L^1_1(\bbr^3)$. Hence, there exists a solution $f$ to the
Boltzmann equation \eqref{boltzmann equation in M}.
The initial condition $f_0\in L^1_r(\bbr^3)$ with $r>1+\frac{a}{2}$ comes from
\eqref{estimate of J1}, and nonnegativity of the solution is guaranteed by
Proposition \ref{Prop existence for truncated}. Uniqueness of solutions can be easily
proved by following the calculations given in the proof of Lemma \ref{Lem Cauchy sequence}.
The initial condition of the theorem is rather strong in the sense that $L^1_1$ solutions
are obtained from $L^1_r$ initial data for $r>1+\frac{a}{2}$, hence the proof of uniqueness
is much easier than that of \cite{MW99}.
This completes the proof of the theorem.
\end{proof}

\subsection{Hard potential case in the Robertson-Walker spacetime}
\label{Sec hard potential in RW}
In this part we extend the existence result of Theorem \ref{Thm main 1} to the
Robertson-Walker case. The argument is basically the same as in the Minkowski
case. We first truncate a certain part of the collision kernel, then existence
for the truncated equation is guaranteed by Proposition
\ref{Prop existence for truncated}. We obtain a sequence of solutions to the
truncated equations, and by showing that the sequence is a Cauchy sequence we
obtain a solution.

To deal with the Boltzmann equation \eqref{boltzmann equation in RW} in the
Robertson-Walker spacetime, we use the covariant variable $v$. In the
Minkowski case we have $p^0=\langle p\rangle$, so there is no difference
between the following quantities:
\[
\int f(t,p)(p^0)^r\,dp=\int f(t,p)\langle p\rangle^r\,dp.
\]
However, in the Robertson-Walker case we have
\[
v^0=\sqrt{1+R^{-2}(t)|v|^2}\quad\mbox{and}\quad \langle v\rangle =\sqrt{1+|v|^2}.
\]
If we introduce a new quantity:
\[
|f(t)|_{1,r}:=\int f(t,v)(v^0)^r\,dv,
\]
then we have the following relation:
\begin{equation}\label{equivalence between norms}
|f(t)|_{1,r}\leq\|f(t)\|_{1,r}\leq R^r(t)|f(t)|_{1,r},
\end{equation}
where we used
\begin{equation}\label{equivalence between weight functions}
v^0\leq\langle v\rangle\leq R(t)v^0.
\end{equation}
Note that
\[
|f(0)|_{1,r}=\|f(0)\|_{1,r},
\]
since we assume $R(0)=1$.

The main goal of this section is to extend Theorem \ref{Thm main 1} to the
Robertson-Walker case. We will show that the lemmas in the previous section
can be applied to the Robertson-Walker case. To do this we first use $v^0$ as
a weight function, i.e. we estimate $|f(t)|_{1,r}$, because all the calculations
in the Minkowski case are naturally extended to the Robertson-Walker case when
using $v^0$ instead of $\langle v\rangle$, for instance we can use the
energy-momentum conservation. Then, by the relation
\eqref{equivalence between norms} we obtain estimates for $\|f(t)\|_{1,r}$.
Similarly to the Minkowski case we modify the Boltzmann equation
\eqref{boltzmann equation in RW} as
\[
\partial_tf_m=Q_m(f_m,f_m)
:=R^{-3}\iint v_{\phi,m}(g_m)^a\sigma_{0,m}(\omega)(f'_mf'_{m*}-f_mf_{m*})\,
d\omega\,du,
\]
where
\[
v_{\phi,m}=\frac{\min\{g\sqrt{s},m\}}{v^0u^0},\quad
g_m=\min\{g,m\},\quad
\sigma_{0,m}=\min\{\sigma_0(\omega),m\}.
\]
Then, the truncated equation has a unique global solution by Proposition
\ref{Prop existence for truncated}.
The following lemmas show that Lemma \ref{Lem uniform boundedness of Y_r} and
\ref{Lem Cauchy sequence} can be extended to the Robertson-Walker case.
\begin{lemma}\label{Lem uniform boundedness of Y_r in RW}
For any $r\geq 0$ and $T>0$, there exists a constant $C_r$ which does not
depend on $m$
such that if $\|f_0\|_{1,r}$ is bounded, then
\[
\sup_m\sup_{t\in[0,T]}|f_m(t)|_{1,r}+\|f_m(t)\|_{1,r}\leq C_r.
\]
\end{lemma}
\begin{proof}
We first estimate $|f_m(t)|_{1,r}$, and then obtain the desired result by using
the relation \eqref{equivalence between norms}.
By Proposition \ref{Prop existence for truncated}
and the relation \eqref{equivalence between norms} we have
\[
\sup_{t\in[0,\infty)}|f_m(t)|_{1,r}\leq C\quad\mbox{for}\quad 0\leq r\leq 1,
\]
and now assume $r>1$.
In the Robertson-Walker case,
$v^0$ depends on time and decreases as time evolves for each $v$. To be precise,
\[
v^0=\sqrt{1+R^{-2}(t)|v|^2}\quad\mbox{and}\quad
\partial_tv^0=-\frac{\dot{R}(t)}{R^3(t)}\frac{|v|^2}{v^0}\leq 0,
\]
since we assume $R(t)\geq 1$ and $\dot{R}(t)\geq 0$.
If we follow the calculation of the proof of Lemma
\ref{Lem uniform boundedness of Y_r},
then we obtain
\begin{align*}
&\frac{d}{dt}|f_m(t)|_{1,r}\\
&=\frac{R^{-3}(t)}{2}\iiint v_{\phi,m}(g_m)^a\sigma_{0,m}(\omega)f_mf_{m*}
((v'^0)^r+(u'^0)^r-(v^0)^r-(u^0)^r)\,d\omega\,du\,dv\\
&\quad+\int f_m(t,v)\frac{\partial}{\partial t}\Big[(v^0)^r\Big]\,dv,
\end{align*}
and the second integral is negative.
Hence, we may only consider
\begin{align*}
&\frac{d}{dt}|f_m(t)|_{1,r}\\
&\leq\frac{R^{-3}(t)}{2}\iiint v_{\phi,m}(g_m)^a\sigma_{0,m}(\omega)f_mf_{m*}
((v'^0)^r+(u'^0)^r-(v^0)^r-(u^0)^r)\,d\omega\,du\,dv,
\end{align*}
and follow the same calculations of Lemma \ref{Lem uniform boundedness of Y_r}
to obtain
\[
\sup_m\sup_{t\in[0,T]}|f_m(t)|_{1,r}\leq C_r.
\]
Consequently, the relation \eqref{equivalence between norms} gives the desired
result.
\end{proof}
\begin{lemma}\label{Lem Cauchy sequence in RW}
Consider the sequence $\{f_m\}$ on any finite time interval $[0,T]$. For any
small number $\delta>0$, there exists a positive number $M$ such that if
$k,m\geq M$, then
\[
\sup_{t\in[0,T]}\|f_k(t)-f_m(t)\|_{1,1}\leq\delta.
\]
\end{lemma}
\begin{proof}
Similarly to the previous lemma, we first estimate $|f_k-f_m|_{1,1}$ and then
obtain the desired result by using \eqref{equivalence between norms}.
By direct calculation, we have
\begin{align*}
&\frac{d}{dt}|f_k(t)-f_m(t)|_{1,1}=\frac{d}{dt}\int|f_k(t,v)-f_m(t,v)|v^0\,dv\\
&=\int\partial_t\Big[|f_k(t,v)-f_m(t,v)|\Big]v^0
+|f_k(t,v)-f_m(t,v)|\partial_t v^0\,dv\\
&=\int\mbox{sgn}(f_k-f_m)(Q_k(f_k,f_k)-Q_m(f_m,f_m))v^0\,dv\\
&\quad-\frac{\dot{R}(t)}{R^3(t)}\int|f_k(t,v)-f_m(t,v)|\frac{|v|^2}{v^0}\,dv\\
&\leq\int\mbox{sgn}(f_k-f_m)(Q_k(f_k,f_k)-Q_m(f_m,f_m))v^0\,dv.
\end{align*}
Hence, we can follow the proof of Lemma \ref{Lem Cauchy sequence} and obtain
a positive number $M$ such that if $k,m\geq M$, then
\[
\sup_{t\in[0,T]}|f_k(t)-f_m(t)|_{1,1}\leq \delta.
\]
Consequently, we obtain the desired result by applying
\eqref{equivalence between norms}, and this completes the proof.
\end{proof}
By the same argument as in Theorem \ref{Thm main 1}, we obtain the following
theorem.
\begin{theorem}\label{Thm main 2}
Suppose that the scattering kernel has the form of \eqref{scattering kernel}
and initial data satisfies $f_0\in L^1_r(\bbr^3)$ for $r>1+\frac{a}{2}$ with
$f_0\geq 0$. Then, the Boltzmann equation \eqref{boltzmann equation in RW}
in the Robertson-Walker spacetime has a unique global solution
$f\in C([0,\infty);L^1_1(\bbr^3))$
with $f(t)\geq 0$.
\end{theorem}

\section{Summary and outlook}

In this paper global existence theorems have been proved for spatially
homogeneous solutions of the Boltzmann equation in Minkowski space and in
spatially flat Robertson-Walker spacetimes. This was done for a class of
collision kernels of hard potential type and these theorems extend existing
results for smooth collision kernels with compact support to cases which are
closer to those which naturally arise in physical problems.

There are many directions in which this work might be generalized. Do the
special relativistic solutions constructed here converge to equilibrium as
$t\to\infty$? If so, can a useful analogue be proved for the Robertson-Walker
spacetimes? Note that in the latter case equilibrium is impossible since the
existence of an equilibrium solution of the Boltzmann equation in a spacetime
implies the existence of a timelike Killing vector field (\cite{israel},
p. 1167) and vector fields of this kind do not exist in most Robertson-Walker
spacetimes. It is nevertheless the case that in practise ideas related to
equilibrium matter distributions are used in cosmology and so it should be
possible to formulate some mathematical analogue of equilibrium solutions in
an expanding cosmological model.

Now that a global existence result has been obtained in a class of homogeneous
and isotropic cosmological models it is natural to ask whether a similar
result can be proved in general homogeneous cosmological models which expand
for ever. In this paper the device of writing the equation in terms of the
covariant components of the momentum is used to simplify the equation.
In fact this trivializes the kinetic part of the equation. The same trick
would work in more general models of Bianchi type I but not in general
homogeneous models. It is nevertheless the case that this transformation
does produce some simplification in general and this has been exploited
in the study of the late-time behaviour of the Einstein-Vlasov system in
\cite{nungesser}.

Another generalization is to look at the global existence question for
homogeneous solutions of the Einstein-Boltzmann system. A template for this
could be provided by the results of this type for the Einstein equations
coupled to other matter models proved in \cite{rendall95}. For the class of
collision kernels considered here there are no local existence theorems
available for inhomogeneous solutions of the Einstein-Boltzmann system. There
are global existence theorems available for the Boltzmann equation with this
type of kernel in special relativity with small or close to homogeneous
initial data. The small data problem for the Einstein-Boltzmann system is out
of reach at present since even the corresponding problem for the
Einstein-Vlasov system has not been solved. On the other hand in the presence
of a positive cosmological constant there are global existence results for
solutions of the Einstein-Vlasov system evolving from data which are close to
homogeneous \cite{ringstrom}.

It is well-known that the Boltzmann equation is not well-posed in the past time
direction. In cosmology this equation is of interest for the very early
universe and so it is natural to enquire if solutions of the
Einstein-Boltzmann system can be constructed which extend all the way back to
the big bang. It is not reasonable to construct these by evolving backwards
in time and a possible alternative is to pose data at the singularity. A
formal study of this type of procedure has been carried out in \cite{tod}
but corresponding existence proofs have not yet been developed. In conclusion,
the study of the Boltzmann equation in curved spacetimes and its coupling
to the Einstein equations gives rise to a variety of challenging mathematical
problems.

\end{document}